\documentclass[aip,jmp,12pt]{revtex4-1}
\usepackage{amsmath,amsfonts,amssymb,amsthm,url}
\newcounter{theorem}

\numberwithin{equation}{section}
\newtheorem{thm}[theorem]{Theorem}        %
\newtheorem*{thm*}{Theorem}
\newtheorem{lemma}[theorem]{Lemma}

\newtheorem{cor}[theorem]{Corollary}      %
\newtheorem{defnn}[theorem]{Definition}
\newtheorem{assump}[theorem]{Assumption}  %

\newcommand{\dist}{{\, \rm dist}}

\newcommand{\com}[2]{ \left [ #1 \, ,\, #2 \right ] }

\newcommand{\ds}{\displaystyle}

\newcommand{\be}{\begin{equation}}
\newcommand{\ee}{\end{equation}}
\newcommand{\ba}{\begin{eqnarray}}
\newcommand{\ea}{\end{eqnarray}}
\newcommand{\im}{\mathrm i}

\newcommand{\1}{{\mathbf 1}}
\newcommand{\R}{\mathbb R}
\newcommand{\N}{\mathbb N}
\newcommand{\C}{\mathbb C}

\newcommand{\di}{\mathrm d}

\theoremstyle{definition}

\begin{document}

\title{A Note on the Switching Adiabatic Theorem}

\author{Alexander Elgart}
\email{aelgart@vt.edu}
\author{George A. Hagedorn}
\email{hagedorn@math.vt.edu}
\affiliation{Department
of Mathematics, Virginia Tech., Blacksburg, VA, 24061 }
\date{Revision date: 16 Jun 2012}

\begin{abstract}
We derive a nearly optimal upper bound on the running time
in the adiabatic theorem for a switching family of Hamiltonians.
We assume the switching Hamiltonian is in the Gevrey class $G^\alpha$
as a function of time, and we show that the error in adiabatic
approximation remains small for running times of order
$g^{-2}\,|\ln\,g\,|^{6\alpha}$.
Here $g$ denotes the minimal spectral gap between the eigenvalue(s)
of interest
and the rest of the spectrum of the instantaneous Hamiltonian.
\end{abstract}

\maketitle
%
%
%
%
%%%%%%%%%%%%%%%%%%%%%%%%%%%%%%%%%%%%%%%%%%%%%%%%%%%%%%%%%%%%%%%%%%%%%%%%%%%%%%%%
%-------------------------------------------------------------------------------
%
%-------------------------------------------------------------------------------
%%%%%%%%%%%%%%%%%%%%%%%%%%%%%%%%%%%%%%%%%%%%%%%%%%%%%%%%%%%%%%%%%%%%%%%%%%%%%%%%
%
%
\section{Introduction}\label{sec:intro}
%%%%%%%%%%%%%%%%%%%%%%%%%%%%%%%%%%%%%%%%%%%%%%%%%%%%%%%%%%%%%%%%%%%%%%
%In this paper w

We consider the dynamical behavior of a quantum system
governed by a time dependent Hamiltonian $H(t)$,
characterized by the following properties:\\
(a) $H(t)$ is a smooth family of self--adjoint, bounded operators, and\\
(b) the time derivative $\dot H(t)$ is compactly supported
in the interval $[0,\tau]$.\\
One can think of such a family as a switching system,
{\it i.e.}, a system that coincides with $H(t)=H_I$ in the past ($t\le 0$),
and switches to the system $H(t)=H_F$ in the future ($t\ge \tau$).

Our goal is to establish an upper bound on the minimal running
time $\tau$ needed to make the error small in the adiabatic theorem
for the switching Hamiltonian $H(s)$.
Our results substantiate the ideas presented in Ref.~\onlinecite{elgart}.

It is convenient to rescale the time $t$ to $s=t/\tau$.
With a minor abuse of notation,
the system then evolves according to the Schr\"odinger equation:
\begin{equation}\label{eq:Sch}
i\,\dot\psi_\tau(s)\ =\ \tau\,H(s)\,\psi_\tau(s)\,,\quad
\psi_\tau(0)\ =\ \psi_I\,,
\end{equation}
where $\psi_I$ is the ground state of $H_I$.
The adiabatic theorem of quantum mechanics ensures that under
certain conditions, if $\psi_\tau(0)$ is close to the ground
state of $H_I$, then $\psi_\tau(1)$ is close to the ground state
of $H_F$ (see the theorems below for details).

%%%%%%%%%%%%%%%%%%%%%%%%%%%%%%%%%
Recently it has been realized that the adiabatic approximation
could be used as the fundamental ingredient for a method of
quantum computation \cite{FGGS}. This adiabatic quantum computing (AQC) 
has generated a resurgence of interest in the adiabatic theorem.
In general, quantum computing attempts to exploit quantum mechanics
to obtain a speedup for classically difficult computational problems. It was subsequently shown in Ref.~\onlinecite{ADKS} that AQC provides a universal model, equivalent (in terms of complexity) to the quantum gate model (and hence to other universal models). 
In adiabatic quantum computing, one solves a computational problem
by using an adiabatically changing Hamiltonian function whose initial
ground state encodes the input and whose final ground state encodes the output.
The time $\tau$ taken to reach the final ground state is
the ``running time'' of the quantum adiabatic algorithm.
One would like to minimize it, while at the same time
keeping the distance between the actual final state
and the desired final ground state small. 
The crucial parameter on which $\tau$ depends
is the minimal value $g$
of the spectral gap $g(s)$ between the ground state energy
of $H(s)$ and the rest of its spectrum.
%%%%%%%%%%%%%%%%%%%%%%

Adiabatic theorems fall into two categories:
those that describe the solutions for {\it all} times,
including times $s\in [0,\,1]$, and those that characterize the solutions
only for large times $s > 1$, where the Hamiltonian is
time--independent again. Interestingly, the latter give more precision
for long times.
We call the first category, the one that applies to all times, ``uniform";
the second is the ``long time" category.

A representative result from the uniform category is the following:
See, {\it e.g.}, Ref.~\onlinecite{AE}:
%%%%%%%%%%%%%%%%%%%%%%%%%%%%%%%%%%%%%%%%%%%%%%%%%%%%%%%%%%%%%%%%%%%%%%
\begin{thm*}[Uniform adiabatic theorem]\label{thm:adi}
Suppose $H(s)$ is a $\tau$--independent, twice differentiable
family of bounded self--adjoint operators on the interval $[0,\,1]$.
Suppose in addition that
\be\label{eq:gap_g}
g\ :=\ \min_{s\in[0,\,1]}\,g(s)\ >\
0 \quad \mbox{ for all }\quad s\in[0,\,1]\,.
\ee
Then, for any $s\in[0,\,1]$, the solution
$\psi_\tau(s)$ to the initial value problem \eqref{eq:Sch} satisfies
\be\label{eq:adi}
\dist\left(\psi_\tau(s),\,Range\ P(s)\right)\ =\ O\left(\tau^{-1}\right)\,,
\ee
where $P(s)$ is the orthogonal projection onto the corresponding eigenstate of $H(s)$.
\end{thm*}
%%%%%%%%%%%%%%%%%%%%%%%%%%%%%%%%%%%%%%%%%%%%%%

A characteristic result from the long time category is
%%%%%%%%%%%%%%%%%%%%%%%
\begin{thm*}[Long time adiabatic theorem]\label{thm:adi_long}
Suppose $H(s)$ is a $\tau$--independent, $C^\infty$ family of bounded
self--adjoint operators that satisfies \eqref{eq:gap_g}.
If $\dot H(s)$ is supported on $[0,\,1]$, then the solution
$\psi_\tau(s)$ to the initial value problem \eqref{eq:Sch} satisfies
\be\label{eq:adi_long}
\dist\left(\psi_\tau(s),\,Range\ P_F\right)\ = \
o\left(\tau^{-n}\right)\quad\mbox{for all}\quad s\ge 1\,,
\ee
for any $n\in\N$.
\end{thm*}
%%%%%%%%%%%%%%%%%%%%%%%%%%%%%%%%%%%%%%%%%%%%%%%
\begin{remarks}
\begin{enumerate}
%\item[]
\item One can summarize these results by saying that slowly starting
and finishing the interpolation decreases the error.
\item In general, there is no uniformity in $n$ in \eqref{eq:adi_long};
the term on the right hand side is of order $c_n\,\tau^{-n}$ where $c_n$
grows rapidly with $n$ ({\it c.f.}, the following  discussion).
\item The distinction between the uniform and long time
adiabatic theorems has an analog for integrals.
Suppose $g(s)\in C^{\infty}({\mathbb R})$ has support in $[0, 1]$.
Then
\[
\int_0^s\,g(t)\ e^{it\tau}\,dt\ =\
\begin{cases}
\ o(\tau^{-n})&\mbox{ if}\quad s\ge 1\, ; \\
\ O(\tau^{-1})&\mbox{ if}\quad s\in(0,\,1)\,.
\end{cases}\]
\end{enumerate}
\end{remarks}
%%%%%%%%%%%%%%%%%%%%%%%%
To describe our result we begin by introducing some notation.
Let $H_I$ and $H_F$ be two self--adjoint operators on a
Hilbert space $\mathcal H$ that satisfy
$\|H_I\|=\|H_F\|=1$.
Let $H(s)$ be a $C^\infty$ family that switches between
$H_I$ and $H_F$ as above, and let $P(s)$ be the orthogonal
projection onto the eigenvalue $E(s)$ of $H(s)$.

We assume the following hypothesis:
%%%%%%%%%%%%%%
\begin{assump}[Minimal gap]\label{MinGap}
For all $s\in[0,\,1]$, we assume the operator $H(s)$
has an eigenprojector $P(s)$,
with eigenenergy $E(s)$ separated by a gap $g(s)$
from the rest of its spectrum.
We assume $\displaystyle g=\min_{s\in[0,\,1]}\,g(s)$ is strictly
greater than zero.
\end{assump}
%%%%%%%%%%%%%%

\begin{remarks}
\begin{enumerate}
%\item[]
\item The eigenprojection $P(s)$ is allowed to be degenerate.
In particular, it can be infinitely degenerate.
\item If ${\rm Rank}\ P_I<\infty$, then it follows from analytic
perturbation theory that\\
${\rm Rank}\ P(s)\,=\,{\rm Rank}\ P_I$ for all $s\in[0,\,1]$.
\end{enumerate}
\end{remarks}

In the quantum computational setting the typical value of the gap $g$ is very small. We are interested in minimizing  the running time $\tau$ so that the error in the adiabatic approximation is small for such values of $g$. We therefore investigate how the
coefficients $c_n$ depend on the gap $g$, and we
minimize the running time $\tau$ so that
$c_n\,\tau^{-n}=o(1)$ for an optimally chosen value of $n$.  
One recent result in this direction is Ref.~\onlinecite{LRH}, which
states that the running time $\tau$ of the order $ g^{-3}$ makes the error in the long time adiabatic theorem small.

It has to be noted that as long as one is interested in just making the error in the adiabatic theorem $o(1)$ in $g\ll1$ rather than say $O(g^n)$, there is not much difference between the uniform and long time adiabatic theorems. 
Our main assertion below is consequently formulated as a uniform adiabatic theorem:

\begin{thm}\label{thm:main}
Under Assumption \ref{MinGap}, the error
$\dist\left(\psi_\tau(s),\,Range\ P(s)\right)$
%\ = \ o(1)\,.
for $s\in[0,1]$
%\ee
in the adiabatic theorem is $o(1)$ for $g\ll 1$, whenever
\[
\tau\ \ge\ K\,g^{-2}\,|\ln\,g\,|^{6\alpha}\,,
\]
for some $g$--independent constant $K > 0$,
if the Hamiltonian belongs to the class $G^\alpha$ with $\alpha>1$,
given in the following definition.
\end{thm}

%%%%%%%%%%%%%%%%%
\begin{defnn}\label{def:Gevrey}
An operator valued function $H(s)$
belongs to the Gevrey class $G^\alpha(R)$, Ref.~\onlinecite{Gevrey},
if $\dot H(s)$ is supported in the interval $[0,\,1]$ and
there exists a constant $C$,
such that for any $k\ge 1$,
\[
\max_{s\in[0,\,1]}\ \left\|\,\frac{d^{k}H(s)}{d^{k}s}\,\right\|\ \le\
C\,R^k\,k^{\alpha\,k}\,.
\]
We define $\displaystyle G^\alpha=\bigcup_{R>0}\,G^\alpha(R)$.
\end{defnn}
%%%%%%%%%%%%%%%%%%
When $\alpha =1$, this class coincides with the set of analytic functions,
and the only such functions are constants.
For $\alpha >1$ there are functions in the class that are not constant.

\subsection{A Prototypical Example: An Interpolating Hamiltonian}
We call $H(s)$ an interpolating Hamiltonian if
\begin{equation}\label{eq:H}
H(s)\,:=\,(1-f(s))\,H_I\,+\,f(s)\,H_F\,,
\end{equation}
where $f$ is a monotone increasing function on $\R$ that satisfies
\be\label{eq:f}
f\in C^\infty(\R)\,;\quad\mbox{supp}\,\dot f\subset[0,1]\,;\quad
f(0)=0\,;\quad\mbox{and}\quad f(1)=1\,.
\ee
Specifically, we can construct $f\in \C^\infty(\R)$ as follows:
\[
f(t)\ =\ \int_{-\infty}^t\,g(s)\,ds\,,\quad\mbox{where}\quad
g(s)\ =\
\begin{cases}
\ 0&\mbox{ if}\quad s\notin[0,\,1]\, ;\\
\ \beta\,\exp{\left(-\ \frac{1}{s(1-s)}\right)}
&\mbox{ if}\quad s\in(0,\,1)\,.
\end{cases}
\]
Here, $\beta$ is a normalization constant,
chosen so that $f(1)=1$.
For this family, we have %\marginpar{\bf Verify; $C=$?}
\[
\left\|\,\frac{d^{k}H%_{I,F}
(s)}{d^{k}s}\,\right\|\ =\
\left|\,\frac{d^{k}f(s)}{d^{k}s}\,\right|\ \|H_F-H_I\|\ \le\
C\,k^{2k}\,.
\]
Hence, $H(s)\in G^2$.
%%%%%%%%%%%%%%%%%%%%%%%%%%%%%%%%%%%%%%%%%%%%%%%%

\begin{remarks}
\begin{enumerate}
\item For analytic families of Hamiltonians one can obtain
sharper control of the transition probability in the adiabatic
approximation. For two level systems this goes back to Landau and Zener,
who showed that the transition probability was $O(e^{-Cg^2\tau})$;
see Refs.~\onlinecite{HLZ,J,JS,JP2} for rigorous treatments.
The analogous statement (albeit with the less explicit dependence on $g$)
is known to hold for the general analytic $H(s)$;
see Ref.~\onlinecite{NS} and references therein.
For the robust adiabatic theorem with
${\rm Rank }\,H_F\ll \dim \mathcal H$,
one can develop a lower bound on the run time of the form
$\tau=O(g^{-2}/|\ln\,g\,|)$, see Ref.~\onlinecite{CE}.
\item The above remark shows that our result is nearly optimal.
In the switching family setting, the requirement that $\dot H(s)$
have compact support exacts a price on the shortest achievable
run time $\tau$.
In our result, it introduces logarithmic corrections to the
Landau--Zener type bound $\tau=o(g^{-2})$.
\item When the gap $g(s)$ becomes small only for finitely many
times $s$, such as in the Grover search problem \cite{G},
one can devise a gap--sensitive interpolating function $f(s)$
that yields a much better estimate, $\tau=O(g^{-1})$,
see {\it e.g.}, Ref.~\onlinecite{JRS} and references therein.
\end{enumerate}
\end{remarks}
%%%%%%%%%%%%%%%%%%%%%%%%
\subsection{Relation with Past Work}
Mathematical analysis of adiabatic behavior has a very rich history
starting with the first rigorous result by Kato \cite{K} for rank one
projections $P(s)$ and Nenciu \cite{N} for more general $P(s)$.
We do not attempt to give an exhaustive survey of the related literature,
but rather focus on articles that steadily improved the understanding of the
long time adiabatic theorem.

The first mathematically rigorous result in this direction goes back to Lenard
\cite{Lenard}. He proved %Theorem \ref{thm:adi_long}
a long time theorem for a Hamiltonian function
of finite rank with no level crossings.
The next significant progress came with the 1987 work of Avron, Seiler, and Yaffe
\cite{asy}. They proved a long time result for a general family of Hamiltonians.
They did not consider dependence on the gap $g$.
In 1991 Joye and Pfister \cite{JP} obtained an estimate on the exponential
decay rate for the $2\times2$ matrix case.
Three years later Martinez \cite{martinez} realized that the adiabatic transition
probability could be considered as a tunneling effect in energy space.
He used microlocal analysis to prove an analogous estimate for the general case.
In 2002, Hagedorn and Joye \cite{HJ} proved exponential error estimates for a
special case using only elementary techniques.
In 2004, Nakamura and Sordoni \cite{NS} combined these tunneling estimates
in energy space with the stationary theory of time--dependent scattering to simplify
the method of Martinez. Both the Martinez and Nakamura--Sordoni results depended
on the choice of a point $E(s)$ in the spectral gap for $H(s)$.
Martinez assumed $E(s)$ was analytic in $s$.
Nakamura and Sordoni assumed it was in the Gevrey class
and satisfied $|E^{(k)}(s)|< C_k\,(1+s^2)^{-\rho}$ for some constants
$C_k$ and $\rho$.
%The adiabatic theorem for
Gevrey class Hamiltonians were also considered in the unpublished works
of Jung \cite{Jung} and Nenciu \cite{N2}.

In this work we propose a straightforward method for computing an upper bound
on the running time, based on Nenciu's expansion, introduced in  Ref.~\onlinecite{nenciu}.
Our method might yield less precise estimates than
those obtained by microlocal analysis.
Our estimate contains the logarithmic prefactor,
but our results do not rely on analyticity of $E(s)$.

We are grateful to the referee for making us aware of the unpublished
preprint \cite{N2}. The purpose of Ref.~\onlinecite{N2} was to quantify the $g$--dependence of the error estimates for the long time adiabatic theorem in Ref.~\onlinecite{nenciu} in the
$g\gg 1$ limit. Although the asymptotics for $g\ll 1$ established
in the present article are different from those of Ref.~\onlinecite{N2}, the basic
strategy of Ref.~\onlinecite{N2} is very similar to ours.

%%%%%%%%%%%%%%%%%%%%%%%%
\subsection{Proof strategy}
The analysis of the wave function $\psi_\tau(s)$ is hampered by the fact that it carries the highly oscillating, memory-dependent phase. In particular, one cannot easily decompose it  into an asymptotic series in powers of $\tau^{-1}$. However,  the {\em projector} $P_\tau(s)$ onto the state   $\psi_\tau(s)$ (given by $P_\tau(s)=|\psi_\tau(s)\rangle\langle\psi_\tau(s)|$) naturally has no phase, and admits the asymptotic expansion of the form  
\be\label{eq:asymp}
P_\tau(s)\ \sim\ B_0(s)\ +\ \frac{1}{\tau}\,B_1(s)\ + \
\frac{1}{\tau^2}\,B_2(s)\ +\ \ldots \,.
\ee
See the next section for details.

A curious fact about this expansion is that it is instantaneous in $H(s)$ -- no memory term is present in any finite order $\tau^{-j}$. The first term in \eqref{eq:asymp} satisfies $B_0(s)=P(s)$. Moreover, $B_j(s)=0$ for all $j\ge1$ provided that {\em all} derivatives of $H(s)$ vanish. In particular, $P_\tau(s)-P(s)=o(\tau^{-n})$ for any $n\in\N$ and $s\ge1$, at least formally.  To make the argument rigorous, one has to estimate the remainder $R_N$ of the series (after the truncation at some $B_N$). It turns out that  $\|R_N\|$ is proportional to $\tau^{-N}$ and depends linearly on  $\dot B_N(t)$, for $t\in[0,s]$
(the remainder encodes the memory of the process).
Therefore one wants to investigate the dependence $\dot B_N(t)$ on the size of the gap $g(t)$ and on the integer $N$. We shall see in Section \ref{sec:B} that in general, $\| \dot B_N(t)\|<C_N g^{-(2N+1)}(t)$, where $C_N$ grows (super) factorially fast in $N$. The proof of the uniform adiabatic theorem, Theorem \ref{thm:adi}, usually involves truncation of the expansion \eqref{eq:asymp} at $N=1$, with the consequent estimate on the error of the form
$R_1=O( \tau^{-1}g^{-3})$.
For a smooth family $H(s)$, however, one can use the fact that the factor
$\tau^{-N}g^{-(2N+1)}$ decreases with $N$ whenever $\tau^{-1}g^{-2}<1$. Therefore one can look for the optimal value $N_{\rm opt}$ for the truncation of the asymptotic series. Evaluation of $\|R_{N_{\rm opt}}\|$ and estimation of the sum of the first $N_{\rm opt}$ terms in the expansion yields Theorem \ref{thm:main} (see Section \ref{sec:togeth} for details). The logarithmic correction in the theorem is an artifact of the appearance of the extra factor $g^{-1}$ in the expression $\tau^{-N}g^{-(2N+1)}$.
%%%%%%%%%%%%%%%%%%%%%%%%%%%%%%%
%%%%%%%%%%%%%%%%%%%%%%%%%%%%%%%
\section{An Asymptotic Expansion for $P_\tau$}
\label{sec:Nenciu}
\setcounter{theorem}{0}

We now consider the initial value problem for the Heisenberg equation,
\be\label{PtauHeis}
\begin{cases}
\ \im\,\dot P_\tau(s)&=\quad\tau\ \com{H(s)}{P_\tau(s)}\\
\ \phantom{\im\,}P_\tau(0)&=\quad P_I
\end{cases}\,.
\ee

In 1993, G.~Nenciu \cite{nenciu} found a general form for the
solution to this Heisenberg equation. His idea
was to look for an asymptotic series of the form \eqref{eq:asymp}. 
Substituting (\ref{eq:asymp}) into (\ref{PtauHeis}), we obtain
a sequence of differential equations
\begin{subequations}\label{eq:nenciuproblem}
\be\label{eq:nenciuproblema}
\im\,\dot B_j(s)\ =\ \com{H (s)}{B_{j+1}(s)},\qquad j=0,\,1,\,\ldots
\ee
In addition, since $P_\tau(s)$ is a projection for each $s$,
we have $P_\tau(s)^2 = P_\tau(s)$. This generates the following sequence
of algebraic relations:
\be\label{eq:nenciuproblemb}
B_j(s)\ =\ \sum_{m=0}^{j}\,B_m(s)\,B_{j-m}(s),\qquad j=0,\,1,\,\ldots
\ee
\end{subequations}
In particular: $B_0(s)^2 = B_0(s)$, so $B_0(s)$ is a projection for
each $s$.

It turns out that the system of hierarchical relations
\eqref{eq:nenciuproblema} and
\eqref{eq:nenciuproblemb} has a unique solution, which is given by the
following recursive construction:
\be\label{eq:defB}
\begin{cases}
\ B_0 (s)&=\quad P(s) \\
\ B_{j}(s)&=\quad\frac{1}{2\pi}\,\int_{\Gamma} R_z(s)
\com{P(s)}{\dot B_{j-1}(s)} R_z (s)\,\di z\ +\ S_j(s)\ -\
2P(s)S_j(s)P(s)\,,
\end{cases}
\ee
where $R_z(s) = (H(s) - z)^{-1}$,
\be\label{eq:defS}
S_j(s)\ =\ \sum_{m=1}^{j-1}\,{B_m(s)\,B_{j-m}(s)} \; ,
\ee
and the contour $\Gamma$ encircles only the ground state
energy. In particular the first order term is given by
\be\label{eq:B1}
B_1(s)\ =\ \frac{1}{2 \pi}\ \int_\Gamma\,R_z(s)\,
\com {P(s)}{\dot P(s)}\,R_z(s)\,\di z \; .
\ee

One can truncate the expansion  \cite{ESn} \eqref{eq:asymp}
at some finite order $k>0$ by
observing that
\be\label{eq:fin_ord}
P_\tau(s)\ =\ B_0(s)  + \frac{1}{\tau} B_1(s)\ +\
\ldots\ +\ \frac{1}{\tau^{N}}\ B_N(s)\  - \ \frac{1}{\tau^{N}}\
\int_0^s\,U_\tau(s,r)\,\dot B_N(r)\,U_\tau(r,s)\,\di r \,.
\ee
where $U_\tau(s,t)$ is the unitary Schr\"odinger propagator that
satisfies
\be
\begin{cases}
\ \im\,\frac{\partial \phantom s}{\partial s}\,U_\tau(s,\,r)&=\quad
\tau\,H(s)\,U_\tau(s,r)\\
\ \phantom{\im\,\frac{\partial \phantom s}{\partial s}}\,
U_\tau(r,\,r)&=\quad\1\,.
\end{cases}
\ee

%%%%%%%%%%%%%%%%%%%%%%%%%%%%%%%%%%%%%%%%%%%%%%%%%%%%%%%%%%%%%%%%%%%%%%%%
\section{Estimates on $B_n$}\label{sec:B}
\setcounter{theorem}{0}
Without loss of generality we assume that the constants $C$, $R$ 
in Definition \ref{def:Gevrey} and the value of $1/g(s)$
are all greater than or equal to $1$.
To estimate the minimal run time $\tau$,
we use the following result:
%%%%%%%%%%%%%%

\begin{lemma}
Suppose $H(s)$ belongs to the Gevrey class $G^\alpha(R)$.
Then,
\be\label{eq:bnd_B}
\left\|\,\frac{d^kB_{n}(s)}{d^ks}\,\right\|\ \le\
\frac{1}{(10\,n+0.3)^2}\ g^{-2n-k}\
\left(\,2\,C\,R\,(k+3n)^{2\alpha}\,\right)^{k+3n}\,.
\ee
\end{lemma}
%%%%%%%%%%%%%%

\begin{proof}
We use induction to prove this lemma.

We define
$$
L(n,k)\ =\ \frac{1}{(10\,n+0.3)^2} \ 
 g^{-2n-k}\ \left(\,2\,C\,R\,(k+3n)^{2\alpha}\,\right)^{k+3n}\,.
$$
Since $B_0(s)=P(s)$, we have $\|B_0(s)\|=1$.
To estimate $\dot B_0(s)$, we use the representation
\be\label{eq:B_0}
B_0(s)\ =\
\frac{1}{2\pi i}\ \int_{\Gamma}\,R_z(s)\,\di z\,,
\ee
where $R_z(s):=(H(s)-z)^{-1}$ and the contour $\Gamma$ is given by
\[
\{z\in \Gamma:\ |z-E_g(s)|=g(s)/2\}\,.
\]
We note that the circumference of $\Gamma$ is $\pi\,g(s)$.

Differentiating both sides of \eqref{eq:B_0} and using
\[
\dot R_z(s)\ =\ -\ R_z(s)\,\dot H(s)\,R_z(s)\,,
\]
we see that
\be\label{eq:dotB_0}
\dot B_0(s)\ =\
-\ \frac{1}{2\pi i}\ \int_{\Gamma}\ R_z(s)\,\dot H(s)\,R_z(s)\,\di z\,.
\ee
Since
\be \label{eq:R_z}
\|R_z(s)\|\ =\ \frac{2}{g(s)}\quad\mbox{for}\quad z\in\Gamma
\ee
and our assumptions require
\be\label{eq:dotH}
\|\dot H(s)\|\ \le\ C\,R\,,
\ee
we obtain
\be\label{eq:dotB_0_bnd}
\|\dot B_0(s)\|\ \le\ C\ \frac{2\,R}{g(s)}\,.
\ee
Recall now the definition of an integer composition:
If $k$ is a positive integer, then a composition $\pi_{k,l}$ of $k$
is an ordered set of positive integers
$p_1,\,p_2,\,\ldots,\,p_l$ whose sum is $k$.
Let $\Pi_k$ denote the set of all possible integer compositions $\pi_{k,l}$
of $k$, and let $\Pi_{k;l}\subset \Pi_k$ denote the set of compositions of
$k$ into exactly $l$ parts.
We note that the cardinality of $\Pi_{k}$ is $2^{k-1}$.

Armed with this notion, we observe that
\be\label{eq:Rrepr}
\frac{d^kR_z(s)}{d^ks}\ = \
\sum_{l=1}^k(-1)^{l}\
\sum_{\pi_{k,l}\in\Pi_{k,l}}c_{\pi_{k,l}}\ \left(
\prod_{p_i\in\pi_{k,l}}\left\{R_z(s)\ \frac{d^{p_i}H(s)}{d^{p_i}s}\right\}\
\right)\ R_z(s)\,,
\ee
where the $c_{\pi_{k,l}}$ are the multinomial coefficients
\be\label{eq:suml'}
c_{\pi_{k,l}}\ = \  {k \choose p_1, \ldots, p_l}
\,.
\ee
In particular,
\be\label{eq:suml}
\sum_{\pi_{k,l}\in\Pi_{k,l}}c_{\pi_{k,l}}\ \le \ l^k
\ee
by the multinomial theorem.

Taking the norm on both sides of \eqref{eq:Rrepr}
and using  equations \ref{eq:R_z} and \ref{eq:suml} as well as the fact that $\alpha>1$,
we obtain  the following  bound\,:
\ba\nonumber
\left\|\,\frac{d^kR_z(s)}{d^ks}\,\right\|
&\le&\sum_{l=1}^k\ \|R_z(s)\|^{l+1}\
\sum_{\pi_{k,l}\in\Pi_{k,l}}\ c_{\pi_{k,l}}\ \prod_{p_i\in\pi_{k,l}}\
\left\|\,\frac{d^{p_i}H(s)}{d^{p_i}s}\,\right\|
\\[3mm]\nonumber
&\le&\ R^k\ \sum_{l=1}^k\ C^l\,\left(\frac{2}{g(s)}\right)^{l+1}
\sum_{\pi_{k,l}\in\Pi_{k,l}}\
c_{\pi_{k,l}}\,\prod_{p_i\in\pi_{k,l}}\ (p_i)^{\alpha p_i}
\\[3mm]\nonumber
&\le&R^k\ k^{\alpha k}\ \sum_{l=1}^k\ C^l\,l^k\,\left(\frac{2}{g(s)}\right)^{l+1}
\\[3mm]\nonumber
&\le&\frac{4}{g(s)}\ \left(\,\frac{2\,C\,R\,k^{2\alpha}}{g(s)}\,\right)^{k}
\\[3mm] \nonumber
%&\le&\frac{2}{g(s)}\ L(0,k)\,.
\ea
This initiates the induction since
\ba\label{eq:iniin}
\|B_0^{(k)}(s)\|\ \le\ g/2\ \sup_{z\in \Gamma}\ \|R_z^{(k)}(s)\| \ &\le& 2\,\left(\,\frac{2\,C\,R\,k^{2\alpha}}{g(s)}\,\right)^{k}
\\[3mm] &<& L(0,k)\,.\nonumber
\ea
Next, we verify the induction step:
Suppose
$$
\left\|\,\frac{d^kB_{n}(s)}{d^ks}\,\right\|\ \le\ L(n,\,k),
$$
for some $n$ and all $k$.

For $z\in\Gamma$,
we use definition \eqref{eq:defB} and the Leibniz rule to bound
\ba\nonumber
&&\left\|\,\frac{d^kB_{n+1}(s)}{d^ks}\,\right\|
\\[3mm]\nonumber
&\le&g(s)\
\sum_{\substack{k_1+k_2+\\ k_3+k_4=k}}\,
\frac{k!}{k_1!\,k_2!\,k_3!\,k_4!}\
\left\|\,\frac{d^{k_1}R_z(s)}{d^{k_1}s}\,\right\|\
\left\|\,\frac{d^{k_2}B_{0}(s)}{d^{k_2}s}\,\right\|\
\left\|\,\frac{d^{k_3+1}B_{n}(s)}{d^{k_3+1}s}\,\right\|\
\left\|\,\frac{d^{k_4}R_z(s)}{d^{k_4}s}\,\right\|
\\[3mm]\nonumber
&&+\quad 2\ \sum_{k_1+k_2+k_3=k}\ \frac{k!}{k_1!\,k_2!\,k_3!}\
\left\|\,\frac{d^{k_1}B_0(s)}{d^{k_1}s}\,\right\|\
\left\|\,\frac{d^{k_2}S_{n+1}(s)}{d^{k_2}s}\,\right\|\
\left\|\,\frac{d^{k_3}B_{0}(s)}{d^{k_3}s}\,\right\|
\\[3mm]\label{eq:bndBN}
&&+\quad\left\|\,\frac{d^kS_{n+1}(s)}{d^ks}\,\right\|\,.
\ea
Using \eqref{eq:defS} and the induction estimate \eqref{eq:bnd_B}, we bound
\ba\label{eq:Sbn} \nonumber
\left\|\,\frac{d^kS_{n+1}(s)}{d^ks}\,\right\|&\le&
\sum_{i=1}^n\ \sum_{k_1+k_2=k}\ \frac{k!}{k_1!\,k_2!}\
\left\|\,\frac{d^{k_1}B_i(s)}{d^{k_1}s}\,\right\|\
\left\|\,\frac{d^{k_2}B_{n+1-i}(s)}{d^{k_2}s}\,\right\|
\\[3mm]\nonumber
&\le&\ g^{-2(n+1)-k}\ \left(\,2\,C\,R\,\right)^{k+3(n+1)}
\\[3mm]\nonumber 
&&\hspace{-1cm}\times\ \sum_{i=1}^n\
\sum_{k_1+k_2=k}\ \frac{k!}{k_1!\,k_2!}\
\frac{\left(k_1+3i\right)^{2\alpha\,(k_1+3i)}}{(10\,i+0.3)^2}\
\frac{\left(k_2+3(n+1-i)\right)^{2\alpha\,(k_2+3(n+1-i))}}
{(10\,n+10.3-10\,i)^2}
\\[3mm]\nonumber  
&<&\frac{0.05}{(10\,n+10.3)^2}\ g^{-2(n+1)-k}\
\left(\,2\,C\,R\,\right)^{k+3(n+1)}\
\left(k+3(n+1)\right)^{2\alpha\,(k+3(n+1))}
\\[3mm]\nonumber
&=& 0.05\cdot L(n+1,\,k) \,,
\ea
where, in the last inequality, we have used Corollary \ref{cor:aux2}.

We use this bound together with \eqref{eq:iniin} 
to verify that
\ba\nonumber
&&2\ \sum_{k_1+k_2+k_3=k}\ \frac{k!}{k_1!\,k_2!\,k_3!}\
\left\|\,\frac{d^{k_1}B_0(s)}{d^{k_1}s}\,\right\|\
\left\|\,\frac{d^{k_2}S_{n+1}(s)}{d^{k_2}s}\,\right\|\
\left\|\,\frac{d^{k_3}B_{0}(s)}{d^{k_3}s}\,\right\|
\\[3mm] \nonumber
&\le&\frac{0.4}{(10\,n+10.3)^2}\ g^{-2(n+1)-k}\
\left(\,2\,C\,R\,\right)^{k+3(n+1)}
\\[3mm]\nonumber
&&\hspace{1cm}\times\quad\sum_{k_1+k_2+k_3=k}\
\frac{k!}{k_1!\,k_2!\,k_3!}\ (k_1)^{2\alpha\,k_1}\
\left(k_2+3(n+1)\right)^{2\alpha\,(k_2+3(n+1))}\
(k_3)^{2\alpha\,k_3}
\\[4mm]\nonumber
&\le&\frac{0.4}{(10\,n+10.3)^2}\  g^{-2(n+1)-k}\
\left(\,2\,C\,R\,\right)^{k+3(n+1)}\
\left(k+3(n+1)\right)^{2\alpha\,(k+3(n+1))}
\\[3mm]\label{eq:Ssan}
&=&0.4\cdot L(n+1,\,k)\,,
\ea
where, in the last inequality, we have used Lemma \ref{lem:aux1}.
%We can now use the multinomial formula
%%
%\[
%(x_1 + x_2 + \cdots + x_m)^k\ =\
%\sum_{k_1+k_2+\cdots+k_m=n}\ \frac{n!}{k_1!\,k_2!\,\cdots\,k_m!}\
%\prod_{1\le i\le m}\,x_{i}^{k_{i}}\,,
%\]
%%

Next, we bound the first contribution in \eqref{eq:bndBN}.

We first observe that for $n=k=0$, we have 
\ba\nonumber
g\ \left\|R_z(s)\right\|\
\left\|B_{0}(s)\right\|\
\left\|\,\frac{dB_{0}(s)}{ds}\,\right\|\
\left\|R_z(s)\right\|
&=& 
4\  g^{-2}\ \left(\,2\,C\,R\,\right)
\\[3mm]\label{eq:firsttt}
&<&0.02\cdot L(1,\,0)\,.
\ea
For $n=0$ and $k=1$, we have 
\ba\nonumber
&&
g\ \sum_{\substack{k_1+k_2+\\ k_3+k_4=1}}\,
\frac{k!}{k_1!\,k_2!\,k_3!\,k_4!}\
\left\|\,\frac{d^{k_1}R_z(s)}{d^{k_1}s}\,\right\|\
\left\|\,\frac{d^{k_2}B_{0}(s)}{d^{k_2}s}\,\right\|\
\left\|\,\frac{d^{k_3+1}B_{n}(s)}{d^{k_3+1}s}\,\right\|\
\left\|\,\frac{d^{k_4}R_z(s)}{d^{k_4}s}\,\right\|
\\[3mm]\nonumber
&&\hspace{2cm}< \ 2^5\ 2^{2\alpha}\ g^{-3}\ \left(\,2\,C\,R\,\right)^{2}
\\[3mm]\label{eq:firstta}
&&\hspace{2cm}< \ 0.001\cdot L(1,\,1)\,.
\ea

For $2n+k\ge2$, we bound
\ba\nonumber
&&
\sum_{\substack{k_1+k_2+\\ k_3+k_4=k}}\,
\frac{k!}{k_1!\,k_2!\,k_3!\,k_4!}\
\left\|\,\frac{d^{k_1}R_z(s)}{d^{k_1}s}\,\right\|\
\left\|\,\frac{d^{k_2}B_{0}(s)}{d^{k_2}s}\,\right\|\
\left\|\,\frac{d^{k_3+1}B_{n}(s)}{d^{k_3+1}s}\,\right\|\
\left\|\,\frac{d^{k_4}R_z(s)}{d^{k_4}s}\,\right\|
\\[3mm]\nonumber
&\le&\frac{2^5}{(10\,n+0.3)^2}\ g^{-2(n+1)-k}\
\left(\,2\,C\,R\,\right)^{k+1+3n}
\\[3mm]\nonumber
%&&\hspace{.5cm}\times\quad\sum_{k_1+k_2+k_3=k}\
%\frac{k!}{k_1!\,k_2!\,k_3!}\
&&\hspace{.5cm}\times\quad\sum_{\substack{k_1+k_2+\\ k_3+k_4=k}}\
\frac{k!}{k_1!\,k_2!\,k_3!\,k_4!}\
(k_1)^{\alpha\,k_1}\ (k_2)^{\alpha\,k_2}\ \left(k_3+1+3n\right)^{2\alpha\,(k_3+1+3n)}\ (k_4)^{\alpha\,k_4}
\\[3mm]\nonumber
&\le&2^5\ \frac{0.3}{(10\,n+0.3)^2}\   g^{-2(n+1)-k}\ \left(\,2\,C\,R\,\right)^{k+1+3n}\
\left(k+1+3n\right)^{2\alpha\,(k+1+3n)}
\\[3mm]\label{eq:firstt}
&<&0.5\cdot L(n+1,\,k)\,.
\ea
where we have again used Lemma \ref{lem:aux1}.

Combining \eqref{eq:bndBN} -- \eqref{eq:firstt} we arrive at
$$
\left\|\,\frac{d^kB_{n+1}(s)}{d^ks}\,\right\|\ \le\ L(n+1,\,k).
$$
This proves the lemma.\end{proof}
%%%%%%%%%%%%%%%%%%%%%%%%%%%%%%%%%%%%%%%%%%%%%%%%%%%%%%%%%%%%%%%%%%%%%%%%

\section{Putting Everything Together}\label{sec:togeth}
\setcounter{theorem}{0}
To estimate the error $\|P_\tau(s)-P(s)\|$,
we use \eqref{eq:fin_ord}, which yields 
\be\label{eq:puttog}
\|\,P_\tau(s)-P(s)\,\|\ \le \ \sum_{j=1}^N \frac{1}{\tau^j}\  \|\, B_j(s)\,\|\,+\, 
\frac{1}{\tau^N}\ \int_0^1\ \|\,\dot B_N(r)\,\|\ \di r\,.
\ee
To find the optimal value for $N$, we estimate the last term first. Using the bound \eqref{eq:bnd_B}, we see that
\ba\nonumber
\frac{1}{\tau^N}\ \int_0^1\ \|\,\dot B_N(r)\,\|\ \di r&<&
\frac{1}{\tau^N}\ g^{-2N-1}\ \left(2\,C\,R(1+3N)^{2\alpha}\right)^{1+3N}
\\[3mm]\label{eq:optimiz}
&=&(\tau/g)^{1/3}\
\left(\,\frac{2\,C\,R\,(3N+1)^{2\alpha}}{(\tau\,g^2)^{1/3}}\,\right)^{3N+1}\,.
\ea
We want to find the value of $N$ that minimizes the right hand side.
Differentiating, we see that the minimizing value of $N$ satisfies
\[
\left\lceil \left(\frac{(\tau\,g^2)^{1/3}}{2\,C\,R}\right)^{1/2\alpha}\ e^{-1} \right\rceil \ \ge \ 3\,N_{\rm opt}+1\ \ge\ \left\lfloor\left(\frac{(\tau\,g^2)^{1/3}}{2\,C\,R}\right)^{1/2\alpha}\ e^{-1}\right\rfloor\,.
\]
Substituting $N_{\rm opt}$ into \eqref{eq:optimiz},
we get the following bound
\[
\frac{1}{\tau^{N_{\rm opt}}}\ \int_0^1\ \|\,\dot B_{N_{\rm opt}}(r)\,\|\ \di r\ \le\
(\tau/g)^{1/3}\
\exp\,\left(\,-\ 2\,e^{-1}\,\alpha\,
\left(\,\frac{(\tau\,g^2)^{1/3}}{2\,C\,R}\,\right)^{1/2\alpha}\,\right)\,.
\]
For $g\ll 1$, the above expression is $o(1)$ for small $g$
provided $\tau$ satisfies
\be\label{eq:tau_opt}
\tau\ \ge\ K\,g^{-2}\,|\ln\,g\,|^{6\alpha}\,,
\ee
for some sufficiently large, $g$--independent constant $K > 2$. 

To estimate the sum of the  first $N=N_{\rm opt}$ terms on the right hand side of \eqref{eq:puttog} we again use the bound \eqref{eq:bnd_B} to get
\begin{eqnarray*}&&\sum_{j=1}^{N_{\rm opt}}  \frac{1}{\tau^j}\ \|\, B_j(s)\,\| \ \le \ \sum_{j=1}^{N_{\rm opt}} \frac{1}{(\tau\,g^{2})^j}\
\left(\,2\,C\,R\,(3j)^{2\alpha}\,\right)^{3j}\\
&&\hspace{.5cm}\le \ \sum_{j=1}^{N_{\rm opt}}\frac{1}{(\tau\,g^{2})^j}\
\left(\,2\,K^{-1/3}\,C\,R\,(3j)^{2\alpha}\,\right)^{3j}\\
&&\hspace{.5cm} \le \ \sum_{j=1}^{N_{\rm opt}^{1/2}}\frac{1}{(\tau\,g^{2})^j}\ \left(\,K^{-1/3}\,\sqrt{2\,C\,R}(\tau\,g^{2})^{1/6}\,\right)^{3j}\,+\,\sum_{j=N_{\rm opt}^{1/2}}^{N_{\rm opt}}\frac{1}{(\tau\,g^{2})^j}\
\left(\,K^{-1/3}\,(\tau\,g^{2})^{1/3}\,\right)^{3j}\\
&&\hspace{.5cm} = \ \sum_{j=1}^{N_{\rm opt}^{1/2}}\left(\,K^{-1/3}\,\sqrt{2\,C\,R}(\tau\,g^{2})^{-1/6}\,\right)^{3j}\,+\,\sum_{j=N_{\rm opt}^{1/2}}^{N_{\rm opt}} \
K^{-j} \\ &&\hspace{.5cm}\le \  2\,\left(K^{-1}\,\left(2\,C\,R\right)^{3/2}\,(\tau\,g^{2})^{-1/2}\,+\,K^{-N_{\rm opt}^{1/2}}\right)\ = \ o(1)\,,\end{eqnarray*}
for $\tau$ that satisfies \eqref{eq:tau_opt} and $g$ sufficiently small.  This proves our main result, Theorem \ref{thm:main}.

\qed
\begin{acknowledgments} This work was partially supported by National Science Foundation grants DMS--0907165 and DMS--1210982. We are grateful to the referee for useful remarks and corrections.
\end{acknowledgments}
%
%%%%%%%%%%%%%%%%%%%%%%%%%
\section{Appendix}
\setcounter{theorem}{0}
In this appendix, we collect several technical results.
%%%%%%%%%%%%%%%%%%%%%%%%%
\begin{lemma}\label{lem:aux}
For all $\alpha\ge 1$,
$k_1=0,\,1,\,2,\,\cdots$,
$k_2=0,\,1,\,2,\,\cdots$, and
$n=1,\,2,\,3,\,\cdots$,
%$k\in \N$ and $n\in\Z_+$,
%We have the bound
%
\be\label{eq:auxbnd1}
\frac{k!}{(k_1!)\,(k_2!)}\
(k_1+n)^{\alpha\,(k_1+n)}\ k_2^{\alpha\,k_2}\ \le\
4^{-(\alpha-1)\min(k_1+n,\,k_2)}\ (k+n)^{\alpha\,(k+n)}\,,
\ee
where $k=k_1+k_2$.
%
%for all $\alpha\ge 1$, $k\in \N$ and $n\in\Z_+$.

\end{lemma}
%%%%%%%%%%%%%%%%%%%%%%%%%
\begin{proof}
%Since $k\ge k_1$, the derivative of the logarithm of
%$f(n):=(k+n)^{k+n}/(k_1+n)^{k_1+n}$ with respect to $n$ is non--negative.
%is a non--decreasing function
%of $n$ (as can be checked by taking the derivative
%of $\log(f(n))$ with respect to $n$),
%Thus, $f(n)$ is monotone increasing in $n$, and
%we need only prove
%\be\label{hotdog}
%\frac{k!}{(k_1!)\,(k_2!)}\ k_1^{k_1}\ k_2^{k_2}\ \le\
%k^k
%\ee
%whenever $k_1+k_2=k$.
%However, for any such $k_1\in\Z_+$ and $k_2\in\Z_+$ with $k_1+k_2=k$,
%$$
%\frac{k!}{(k_1!)\,(k_2!)}\ k_1^{k_1}\ k_2^{k_2}\ \le\
%\sum_{k_1+k_2=k}\,\frac{k!}{(k_1!)\,(k_2!)}\ k_1^{k_1}\ k_2^{k_2}\ =\
%k^k.
%$$
We first prove the above inequality for $\alpha=1$.
We begin by noting that
% To this end, 
%suppose $k_1\ge 0$ and $k_2\ge 0$ are integers with $k_1+k_2=k$.
%Then
$$
\frac{k!}{(k_1!)\,(k_2!)}\ (k_1+n)^{k_1}\ k_2^{k_2}\ \le\
\sum_{\tilde k_1+\tilde k_2=k}\,\frac{k!}{(\tilde k_1!)\,(\tilde k_2!)}\
(k_1+n)^{\tilde k_1}\ k_2^{\tilde k_2}\ =\
(k+n)^k\,.
$$
Since $(k_1+n)^n\le (k+n)^n$,
the inequality is preserved if we multiply the left hand side by
$(k_1+n)^n$ and right hand side by $(k+n)^n$.
Hence
\be\label{eq:auxbnd}
\max_{\ k_1+k_2=k}\ \frac{k!}{(k_1!)\,(k_2!)}\ (k_1+n)^{k_1+n}\ k_2^{k_2}\ \le \ (k+n)^{k+n}\,.
\ee
This proves the result for $\alpha=1$.

We now consider the case $\alpha>1$.
For $0\le x\le 1$, we have $x(1-x)\le 1/4$. Applying this with
$\ds x=\frac{k_1+n}{k+n}$, we see that
\be\label{dog}
(k_1+n)\ k_2\ \le\ \frac 14\ (k+n)^2.
\ee

Suppose $k_2\le k_1+n$.  From (\ref{dog}), we see that
$(k_1+n)^{k_2}\ k_2^{k_2}\le 4^{-k_2}\ (k+n)^{2k_2}$. We multiply
the left hand side of this inequality by $(k_1+n)^{k_1+n-k_2}$, and we multiply the right hand side by the larger or equal quantity
$(k+n)^{k_1+n-k_2}$.
We conclude that
\be\label{firsthalf}
(k_1+n)^{k_1+n}\ k_2^{k_2}\ \le\ 4^{-k_2}\ (k+n)^{k+n}.
\ee

Similarly, if $k_1+n\le k_2$, then (\ref{dog}) implies
$(k_1+n)^{k_1+n}\ k_2^{k_1+n}\le 4^{-(k_1+n)}\ (k+n)^{2(k_1+n)}$.
We multiply the left hand side of this inequality by
 $k_2^{k_2-k_1-n}$ and the right hand side
by the larger or equal quantity $(k+n)^{k_2-k_1-n}$. We conclude that
\be\label{secondhalf}
(k_1+n)^{k_1+n}\ k_2^{k_2}\ \le\ 4^{-(k_1+n)}\ (k+n)^{k+n}.
\ee

Combining (\ref{firsthalf}) and (\ref{secondhalf}) and raising the
result to the $(\alpha-1)$ power, we see that
$$
(k_1+n)^{(\alpha-1)(k_1+n)}\ k_2^{(\alpha-1) k_2}\ \le \  4^{-(\alpha-1)\min(k_1+n,\,k_2)}\ (k+n)^{(\alpha-1)\,(k+n)}.
$$
This and (\ref{eq:auxbnd}) imply the lemma.
%
%For the general value of $\alpha$, the result follows from  the inequality
%
%\[(k_1+n)^{(\alpha-1)(k_1+n)}\ k_2^{(\alpha-1) k_2} \ \le \  %4^{-(\alpha-1)\min(k_1+n,k_2)}\ (k+n)^{(\alpha-1)\,(k+n)}\]
%
 %and  \eqref{eq:auxbnd}.
\end{proof}
%%%%%%%%%%%%%%%%%%%%%%%%%
%Next we establish
%%%%%%%%%%%%%%%%%%%%%%%%%

\vskip 4mm
\begin{cor}\label{cor:aux1}
Suppose $\alpha$, $k_1$, $k_2$, $k$, and $n$ are as in Lemma \ref{lem:aux}.
For $1\le i\le n$, we have
%Let $1\le i \le n$. For $k_1$, $k_2$, $k$, and $n$ as in Lemma \ref{lem:aux},
%
\ba\nonumber
&&\frac{k!}{k_1!\,k_2!}\
(k_1+n+1-i)^{\alpha\,(k_1+n+1-i)}\ (k_2+i)^{\alpha\,(k_2+i)}
\\[3mm]\label{eq:auxbnd2}
&&\le\quad 4^{-(\alpha-1)\min(k_1+1,\,k_2+1)}\
(k+n+1)^{\alpha\,(k+n+1)}.
\ea
\end{cor}
%%%%%%%%%%%%%%%%%%%%%%%%%
\begin{proof}
Suppose first that $k_1+n+1-i\ge k_2+i$.
Then we have
$k_1+n\ge k_2+2\,i-1\ge k_2+1$.
For $a\ge b> 0$, the function
$g(j):= (a+j)^{\alpha(a+j)}\,(b-j)^{\alpha(b-j)}$
is non-decreasing on the interval $0\le j \le b$. 
We apply this with $a=k_1+n+1-i$, $b=k_2+i$, and
$j=i-1$. The inequality $g(0)\le g(j)$ yields
$$
(k_1+n+1-i)^{\alpha\,(k_1+n+1-i)}\ (k_2+i)^{\alpha\,(k_2+i)}\
\le\ (k_1+n)^{\alpha\,(k_1+n)}\ (k_2+1)^{\alpha\,(k_2+1)}.
$$
%\ba\nonumber
%&&\frac{k!}{k_1!\,k_2!}\ (k_1+n+1-i)^{\alpha\,(k_1+n+1-i)}\
%(k_2+i)^{\alpha\,(k_2+i)}
%\\[3mm]\nonumber
%&&\le\quad\frac{k!}{k_1!\,k_2!}\ (k_1+n)^{\alpha\,(k_1+n)}\
%(k_2+1)^{\alpha\,(k_2+1)}.
%\ea
Since $k_2+1<k_1+n$, 
Lemma \ref{lem:aux} with $k_2+1$ in place of $k_2$  implies
$$
\frac{(k+1)!}{k_1!\,(k_2+1)!}\
(k_1+n)^{\alpha\,(k_1+n)}\ (k_2+1)^{\alpha\,(k_2+1)}\ \le\
4^{-(\alpha-1)(k_2+1)}\ (k+n+1)^{\alpha\,(k+n+1)}.
$$
Therefore,
\ba\nonumber
&&\frac{(k+1)!}{k_1!\ (k_2+1)!}\
(k_1+n+1-i)^{\alpha\,(k_1+n+1-i)}\ (k_2+i)^{\alpha\,(k_2+i)}
\\[3mm]\label{ineq1}
&&\le\quad
4^{-(\alpha-1)(k_2+1)}\ (k+n+1)^{\alpha\,(k+n+1)}.
\ea

Suppose next that $k_1+n+1-i<k_2+i$.
We again use $g(0)\le g(j)$, but with $a=k_2+i$, $b=k_1+n+1-i$
and $j=n-i$.  This yields
$$
(k_1+n+1-i)^{\alpha\,(k_1+n+1-i)}\ (k_2+i)^{\alpha\,(k_2+i)}\ \le\
(k_1+1)^{\alpha\,(k_1+1)}\ (k_2+n)^{\alpha\,(k_2+n)}.
$$
%\ba\nonumber
%&&\frac{k!}{k_1!\,k_2!}\ (k_1+n+1-i)^{\alpha\,(k_1+n+1-i)}\
%(k_2+i)^{\alpha\,(k_2+i)}
%\\[3mm]\nonumber
%&&\le\quad\frac{k!}{k_1!\,k_2!}\ (k_1+1)^{\alpha\,(k_1+1)}\
%(k_2+n)^{\alpha\,(k_2+n)}.
%\ea
We note that $k_1+n+1-i<k_2+i$ implies
$k_1+1<k_2-n+2\,i$. Since $i\le n$, we have $k_1+1<k_2+n$.  We apply Lemma \ref{lem:aux} with
$k_1$ replaced by $k_2$ and $k_2$ replaced by $k_1+1$.
This yields
% the roles of $k_1$ and $k_2$ reversed and obtain
$$
\frac{(k+1)!}{(k_1+1)!\ k_2!}\
(k_2+n)^{\alpha\,(k_2+n)}\ (k_1+1)^{\alpha\,(k_1+1)}\ \le\
4^{-(\alpha-1)\,(k_1+1)}\ (k+n+1)^{\alpha\,(k+n+1)}.
$$ 
Therefore,
\ba\nonumber
&&\frac{(k+1)!}{(k_1+1)!\ k_2!}\
(k_1+n+1-i)^{\alpha\,(k_1+n+1-i)}\ (k_2+i)^{\alpha\,(k_2+i)}
\\[3mm]\label{ineq2}
&&\le\quad  4^{-(\alpha-1)(k_1+1)}\ (k+n+1)^{\alpha\,(k+n+1)}\,.
\ea

Inequalities (\ref{ineq1}) and (\ref{ineq2}) imply (\ref{eq:auxbnd2})
because
$$
\frac{k!}{k_1!\ k_2!}\ \le\
\min\ \left\{\ \frac{(k+1)!}{k_1!\ (k_2+1)!},\
\frac{(k+1)!}{(k_1+1)!\ k_2!}\ \right\}\,.
$$
\end{proof}
%%%%%%%%%%%%%%%%%%%%%%%%%

We also have the following consequence of Corollary \ref{cor:aux1}:  
%%%%%%%%%%%%%%%%%%%%%%%%%
\begin{cor}\label{cor:aux2}
For $\alpha>1$, $k=1,\,2,\,3,\,\cdots$, and $n=0,\,1,\,2,\,\cdots$,
we have 
\ba\nonumber
&&\sum_{i=1}^n\ \sum_{\ k_1+k_2=k}\
\frac{1}{(10\,n+10.3-10\,i)^2\ (10\,i+0.3)^2}\
\\[3mm]\nonumber
&&\hspace{3cm}\times\quad\frac{k!}{k_1!\ k_2!}\
(k_1+3\,(n+1-i))^{\alpha\,(k_1+3(n+1-i)}\ (k_2+3\,i)^{\alpha\,(k_2+3i)} 
\\[3mm]\label{eq:auxbndsum}
&\le&\frac{0.05\cdot4^{-(\alpha-3/2)}}{1-4^{-(\alpha-1)}}\
\frac{(k+3\,(n+1))^{\alpha\,(k+3(n+1))}}{(10\,n+10.3)^2}\,.
\ea
%
%for all $k\in \N$ and $n\in \Z_+$.
\end{cor}
%%%%%%%%%%%%%%%%%%%%%%%%%
\begin{proof}
Using Corollary \ref{cor:aux1}, we estimate the sum over
$k_1+k_2=k$ by
\ba\nonumber
&&\sum_{\ k_1+k_2=k}\
\frac{k!}{k_1!\,k_2!}\
 (k_1+3\,(n+1-i))^{\alpha\,(k_1+3(n+1-i))}\ (k_2+3\,i)^{\alpha\,(k_2+3\,i)}
\\[3mm]\nonumber
&\le&(k+3\,(n+1))^{\alpha(k+3(n+1))}\
\sum_{k_1+k_2=k}\ 4^{-(\alpha-1)\,\min\,(k_1+1,\,k_2+1)}
\\[3mm]\nonumber
&\le&(k+3\,(n+1))^{\alpha\,(k+3(n+1))}\ 2\,\cdot\,4^{-(\alpha-1)}\
\sum_{k_1=0}^\infty\ 4^{-(\alpha-1)\,k_1}
\\[3mm]\nonumber
&=&(k+3\,(n+1))^{\alpha\,(k+3(n+1))}\ 
\frac{4^{-(\alpha-3/2)}}{1-4^{-(1-\alpha)}}.
\ea
%
%\[
%\frac{1}{(n+2-i)^2(i+1)^2}\ (k+n+1)^{\alpha\,(k+n+1)}\
%\frac{4^{-(\alpha-3/2)}}{1-4^{-(\alpha-1)}}\,,
%\]
% 
The result follows from this and the inequality
\be\label{crazy}
\sum_{i=1}^n\ \frac{1}{(10\,n+10.3-10\,i)^2\ (10\,i+0.3)^2}\quad
<\quad\frac{0.05}{(10\,n+10.3)^2}\,.
\ee

For $n=1,\,2,\,3$, we check this inequality by explicit computation.
For $n\ge 4$, we first note that
\be\label{eq:ab}
\frac 1{(a-b)\ b}\ =\ \frac 1a\ \left(\ \frac 1{a-b}\,+\,\frac 1b\ \right).
\ee
We take the square of both sides and use $a=10\,n+10.6$ and $b=10\,i+0.3$.
We then sum over $i$ from $1$ to $n$. Using the $i \leftrightarrow n+1-i$ symmetry, this yields
\begin{eqnarray*}
&&\sum_{i=1}^n\ \frac{1}{(10\,n+10.3-10\,i)^2\ (10\,i+0.3)^2}
\\[3mm]
&=&\frac{2}{(10\,n+10.6)^2}\ \left\{\sum_{i=1}^n\ \frac{1}{(10\,i+0.3)^2}\ +\  
 \sum_{i=1}^n\ \frac{1}{(10\,n+10.3-10\,i)(10\,i+0.3)}\right\}
\\[3mm]
&=&\frac{2}{(10\,n+10.6)^2}\ \sum_{i=1}^n\ \frac{1}{(10\,i+0.3)^2}\ +\  
\frac{4}{(10\,n+10.6)^3}\ \sum_{i=1}^n\ \frac{1}{(10\,i+0.3)}\,,
\end{eqnarray*}
where in the second step we have used \eqref{eq:ab} and the $i \leftrightarrow n+1-i$ symmetry one more time.

The first term on the right is bounded by
$$
\frac{2/100}{(10\,n+10.6)^2}\ \sum_{i=1}^\infty\ \frac 1{i^2}\ =\
\frac{1/100}{(10\,n+10.6)^2}\ \frac{\pi^2}3.
$$
The second term is bounded by
$$
\frac{4/100}{(10\,n+10.6)^2}\ \frac{1}{n+1}\ \sum_{i=1}^n\ \frac 1i.
$$
For $n\ge 4$,
$$
\frac{4}{n+1}\ \sum_{i=1}^n\ \frac 1i\ \le\ \frac{5}{3}.
$$
Combining these results,
\ba\nonumber
\sum_{i=1}^n\ \frac{1}{(10\,n+10.3-10\,i)^2\ (10\,i+0.3)^2}
%\\[3mm]\nonumber
&\le&\frac{(\pi^2+5)/300}{(10\,n+10.6)^2}
\\[3mm]\nonumber
&<&\frac{0.05}{(10\,n+10.6)^2}\,,
\ea
which proves (\ref{crazy}).
\end{proof}

\vskip 4mm
Arguments similar to the proof of Lemma \ref{lem:aux}
generalize to prove the following:
% as in Lemma \ref{lem:aux} and Corollary \ref{cor:aux},
%one  can establish 
%%%%%%%%%%%%%%%%%%%%%%%%%
\begin{lemma}\label{lem:aux1}
For $\alpha>1$, define
$\kappa_\alpha:=\frac{4^{-(\alpha-3/2)}}{1-4^{-(\alpha-1)}}$.
Then  for $k=1,\,2,\,3,\,\cdots$ and\\ $n=0,\,1,\,2,\,\cdots$,
\ba\label{eq:auxbnd'}
&&\hspace{-3cm}\sum_{\ k_1+k_2+k_3=k}\quad
\frac{k!}{k_1!\,k_2!\,k_3!}\
(k_1+n)^{\alpha(k_1+n)}\ (k_2)^{\alpha k_2}\ (k_3)^{\alpha k_3}\ \le \ \kappa_\alpha^2\ (k+n)^{\alpha(k+n)}\,;
\\[3mm]
&&\hspace{-3cm} \sum_{\substack{\ k_1+k_2+k_3\\+k_4=k}}\quad
\frac{k!}{k_1!\,k_2!\,k_3!\,k_4!}\
(k_1+n)^{\alpha(k_1+n)}\ (k_2)^{\alpha k_2}\
(k_3)^{\alpha k_3}\ (k_4)^{\alpha k_4}\ \le\ \kappa_\alpha^3\ (k+n)^{\alpha(k+n)}\,.
\ea
%
%for all $k\in \N$ and $n\in \Z_+$.
\end{lemma}
\end{document}